\documentclass[draftcls,onecolumn]{IEEEtran}
\usepackage{graphicx,cite,calc,dsfont}
\usepackage[defs]{ams}
\usepackage{eqnarray,subeqn,xcolor}
\usepackage{hyperref}
\usepackage{amsthm,amssymb}
\newtheorem{thm}{Theorem}
\newtheorem{lem}{Lemma}
\newtheorem{cor}{Corollary}
\theoremstyle{definition}

\newtheorem{mydef}{Definition}

\begin{document}

\title{Throughput of Large One-hop Wireless Networks with General Fading}

\author{ Seyed Pooya Shariatpanahi, Babak Hossein Khalaj, Kasra Alishahi, Hamed Shah-Mansouri
 }

\maketitle
\thispagestyle{empty}
\pagestyle{empty}

\section*{Abstract}
\emph{\let\thefootnote\relax\footnotetext{  S. P. Shariatpanahi is with the School of Computer Science, Institute for Research in Fundamental Sciences (IPM), Tehran, Iran.  B. H. Khalaj is with the Department of Electrical Engineering and Advanced Communication Research Institute (ACRI), Sharif University of Technology, Tehran, Iran. K. Alishahi is with the Department of Mathematical Sciences, Sharif University of Technology, Tehran, Iran (emails: pooya@ipm.ir, khalaj@sharif.edu, alishahi@sharif.edu, hshahmansour@alumni.sharif.edu).  \\ \\ This research was in part supported by a grant from IPM. } Consider $n$ source-destination pairs randomly located in a shared wireless medium, resulting in interference between different transmissions. All wireless links are modeled by independently and identically distributed (i.i.d.) random variables, indicating that the dominant channel effect is the random fading phenomenon. We characterize the throughput of one-hop communication in such network. First, we present a closed-form expression for throughput scaling of a heuristic strategy, for a completely general channel power distribution. This heuristic strategy is based on activating the source-destination pairs with the best direct links, and forcing the others to be silent. Then, we present the results for several common examples, namely, Gamma (Nakagami-$m$ fading), Weibull, Pareto, and Log-normal channel power distributions. Finally -- by proposing an upper bound on throughput of all possible strategies for super-exponential distributions -- we prove that the aforementioned heuristic method is order-optimal for Nakagami-$m$ fading. }\\ \\
\emph{\textbf{Index Terms}}---General fading, one-hop wireless networks, random connection model, throughput scaling.

\section{Introduction}\label{Sec_Introduction}

Determining the capacity of general wireless networks remains an open problem among information-theorists \cite{Ephremides}, \cite{Goldsmith}. The main difficulty in addressing such problem is finding the optimum way to overcome and benefit from the broadcast nature of the shared wireless medium which in turn results in \emph{interference} phenomenon, further complicating the problem. In fact, in a large wireless network, a vast number of possible transmission management strategies can be adopted, thus, adding to the complexity of finding the optimum one.

The most successful attempts to address such problem are the ones which consider less-rigorously defined capacity measures -- such as network throughput -- as the number of nodes goes to infinity. In such attempts, throughput is defined as the rate at which  all source-destination pairs can communicate, under some restricting assumptions on network operation such that the problem remains tractable.

As notable examples, \cite{Gupta}, \cite{Kulkarni}, \cite{ElGamal}, and \cite{Fran} derive such throughput scaling laws. The authors in \cite{Gupta} show that -- using multi-hop strategy and spatial reuse concept -- one can achieve the aggregate throughput scaling of order $\sqrt{n}$, where $n$ is the number of the nodes in the network. In other words, based on two main themes, namely, multi-hop strategy coupled with the spatial reuse concept, the resulting throughput per node vanishes as $1/\sqrt{n}$ for large networks. In a network with multi-hop strategy, packets hop between neighboring nodes and follow their way towards destination. At each hop, transmitter power is maintained as low as possible just to ensure reliable communication with the neighboring node. Consequently, wireless channel path loss will provide the possibility of spatial reuse where far-enough nodes can remain transmitting at the same time.

The main restricting assumption in \cite{Gupta} is the decode-and-forward strategy at relay nodes, which treats interference as noise. The papers \cite{Kulkarni}, \cite{ElGamal} and \cite{Fran} also derive the same $\sqrt{n}$ aggregate throughput scaling for multi-hop schemes. These papers also treat interference as noise, and arrive at the same scaling as \cite{Gupta}, indicating a fundamental limitation attached to the multi-hop strategy which limits the aggregate throughput to $\sqrt{n}$. In order to overcome such restriction, in \cite{Ozgur1} the authors achieve linear throughput scaling (i.e. of order $n$) by proposing a hierarchical cooperation scheme. Their result indicates non-vanishing throughput per node for large networks. The authors in \cite{Ozgur1} follow two main themes. The first theme is employing distributed Multiple-Input Multiple-Output (MIMO) techniques and therefore, not treating interference as noise. The second theme is designing a careful hierarchical scheme -- by exploiting the spatial reuse concept -- for data distribution among closely located nodes. Further improvements on performance of hierarchical design are also provided in \cite{Ozgur2}, \cite{Ghaderi}, and \cite{Ozgur3}.

The core key concept in all aforementioned papers is the spatial reuse idea, which relies heavily on the path loss phenomenon. In other words, geographical separation between nodes allows the network designer to activate a large number of them simultaneously, without being disturbed by inter-node interference. However, in many situations the dominant factor in channel condition is the random fading, and not the path loss effect (see e.g., \cite{Gowikar}, \cite{Ebrahimi}, \cite{Cui_Opportunistic}, \cite{Cui}, \cite{Ebrahimi_report}, \cite{Ebrahimi_2011}, and \cite{Pooya}). In such networks, any signal transmission from a node, affects all other nodes in the network equally likely. This effect is modeled by assuming that all links in the network are i.i.d. random variables -- leading to the so-called \emph{random connection model}. Consequently, in such model the geometric structure of the network, and the concept of neighbor and distant nodes becomes of less value leading to a model in which every two nodes will be equally neighbors. As a result, such model leads to a critically interference-limited network, and thus, the network throughput will be substantially lower compared with networks in which path loss is the dominant channel characteristic.

It seems a completely different approach is consequently needed to operate such interference-limited class of networks in which the spatial reuse concept can not be adopted. In \cite{Gowikar}, authors assume an ad hoc network under the random connection model. They consider multi-hop strategy, and arrive at the throughput of order $\log(n)$ for a network with Rayleigh fading channels. Comparing their result with those of \cite{Gupta} shows the deteriorating effect of not being able to use the spatial reuse idea. In \cite{Cui_Opportunistic}, the authors arrive at the throughput of the same order of $\log(n)$ for Rayleigh fading networks -- with only two hops -- indicating the fact that under this model, multi-hop strategy is not of high importance. The papers \cite{Ebrahimi} and \cite{Ebrahimi_2011} arrive at the throughput of order $\log(n)$ even for the networks operating under one-hop strategy. Also, if one is allowed to optimize throughput over the class of finite power distributions (an assumption of mainly theoretical importance), then the results in \cite{Cui} indicate the throughput of order $\sqrt{n}$ for two-hop networks. Also, they show that adding more hops does not increase the throughput under such model. They also derive an upper bound of order $n^{1/3}$ for one-hop network operation for networks with optimized channel power distribution. This result combined with the lower bound result of order $n^{1/3}$ for one-hop networks in \cite{Pooya}, shows that the throughput of one-hop networks is of order $n^{1/3}$ when the channel power distribution is optimized. All these papers convey two important messages: First, in networks with Rayleigh fading, we can arrive at the throughput of order $\log(n)$ with just one-hop transmission and adding further hops is not beneficial. Second, even if we are able to optimize the channel power distribution, adding hops beyond the two-hop strategy is not beneficial. Accordingly, we can conclude that in the networks under the random connection model, the main research efforts should be focused on one-hop and two-hop schemes.

\subsection{Problem Model Overview}

In this paper, we assume a one-hop wireless network with $n$ source-destination pairs. We model the power of all links between all sources and destinations by independent and identically distributed (i.i.d.) random variables, with the common probability distribution function (p.d.f.) $f(\gamma)$. This link model is broadly accepted for analyzing the throughput of wireless networks (\cite{Gowikar}, \cite{Ebrahimi}, \cite{Cui_Opportunistic}, \cite{Cui}, \cite{Ebrahimi_report}, and \cite{Ebrahimi_2011}). Furthermore, we do not assume any cooperation among sources and among destinations. Also, destinations do not use interference cancellation techniques. Finally, we consider on-off transmission strategy for each source node.  In other words, each source either can transmit with the maximum allowed power, or should remain silent. The on-off strategy has been shown to be the optimum paradigm in a number of network settings (see e.g., \cite{Abouei} and \cite{Ebrahimi_2011}) while not requiring much feedback from the destination, making it simpler to implement in practice (the papers \cite{Ebrahimi}, \cite{Ebrahimi_report}, \cite{Cui}, and \cite{Pooya} also consider the same set of assumptions). With such assumptions, our problem is finding the throughput of one-hop communication in the network, where throughput denotes the maximum number of possible successful concurrent transmissions. Clearly, the interference phenomenon will limit the throughput (see Fig. \ref{Fig_Schematic}).

\begin{figure}
\begin{center}
\includegraphics[width=0.38\textwidth]{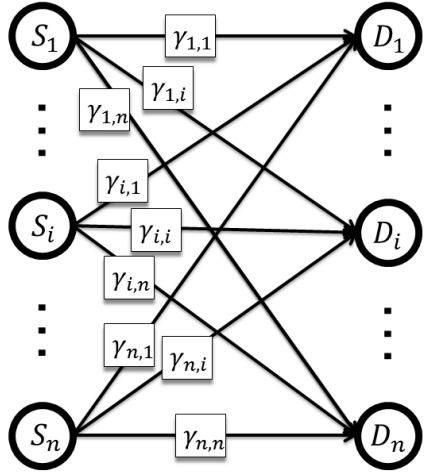}
\end{center}
\caption{Schematic of the network model. $S_i$'s and $D_i$'s are sources and destinations, respectively. The channel power between $S_i$ and $D_j$ is denoted by $\gamma_{i,j}$.\label{Fig_Schematic}}
\end{figure}

\subsection{Our Contributions Versus Previous Results}

The papers on one-hop communication under the same model fail to derive a closed-form throughput expression for the networks with general fading. Accordingly, network throughput in the case of many practical channel power distributions remains unknown. Our goal in this paper is proposing a new analysis framework which addresses this problem in a general way. Specifically, our contributions are as follow:

\begin{itemize}
\item
Showing that the complexity of determining the optimum transmission strategy is at least NP-hard.

\item
Deriving the achievable throughput  of a low-complexity heuristic transmission strategy in closed-form, for a general power distribution $f(\gamma)$.

\item
Characterizing the achievable throughput scaling for networks with Gamma (Nakagami-$m$ fading), Weibull, Log-normal, and generalized Pareto power distributions, as four case studies.

\item
Deriving throughput upper bound for the class of super-exponential power distributions\footnote{By definition, a super-exponential distribution's probability distribution function can be upper bounded by an exponentially-decaying function.}, and proving the order-optimality of the low-complexity heuristic strategy for Nakagami-$m$ fading.

\end{itemize}

In order to further distinguish between our contributions in this paper and earlier results, we consider the most related works to our paper (\cite{Ebrahimi}, \cite{Ebrahimi_report}, \cite{Ebrahimi_2011}, \cite{Pooya}, and \cite{Cui}).

In \cite{Ebrahimi} and \cite{Ebrahimi_2011}, the authors have derived the throughput scaling for the Rayleigh fading case (exponential power distribution). Their scheme is asymptotically equivalent to our scheme, and our Corollary \ref{Cor_Lower_Nakagami} for the Nakagami-$m$ fading covers their result as special case of $m=1$. In \cite{Ebrahimi_report}, they have arrived at the throughput scaling for the Log-normal power distribution as well, which coincides with our Corollary \ref{Cor_Lower_Log_Normal}. In contrast to their work, we follow a completely different approach by using a more powerful mathematical framework, namely, the \emph{intermediate order statistics}\cite{Falk}. Accordingly, we are able to characterize the throughput for a completely general power distribution in closed form. In addition, we prove the order-optimality of the heuristic scheme for Nakagami-$m$ fading, whereas they have proved it for a similar scheme under Rayleigh fading.

In \cite{Pooya}, throughput scaling is only derived for the Pareto power distribution, while here, we arrive at a completely general expression, with a generic fading distribution. In addition, the approach followed in \cite{Pooya} only considers the average throughput, while here we present a convergence in probability result for the throughput, which is a much stronger result. Also, here we discuss the computational complexity of the problem while in \cite{Pooya} no such analysis is provided. Finally, it should be mentioned that in \cite{Pooya}, no upper bounds are presented, while here we arrive at an upper bound for the class of super-exponential distributions. In summary, Corollary \ref{Cor_Lower_Pareto} of this paper completely covers the result in \cite{Pooya} for the Pareto distribution, in a much stronger probabilistic sense, and the other results are sole contributions of this work.

Finally, in \cite{Cui}, the authors investigate the case of optimal power distribution with finite power. Such assumption is not practical in general, as channel power distribution is dictated by the environment  and cannot be optimized. Since here we assume a general -- but fixed and given -- distribution, the result in \cite{Cui} serves as an upper bound for our general throughput expression.

We should also mention that a number of partially-related papers have used the concept of multiuser diversity to arrive at high throughput in networks with randomly varying channels. Their results cannot directly be compared with ours, since their network model (in terms of channel or traffic characteristics) differs from ours. For example, in \cite{Gesbert}, considering cells of $U$ users, a multi-cell network is considered and corresponding power control and user scheduling mechanism are proposed. By considering path-loss in their model, they have improved the previously derived result of order $\log\log(U)$, and reached the throughput of order $\log(U)$. Sohn et al. have arrived at the same result in \cite{Sohn} for a different application. The paper \cite{Kountouris} considers the throughput of two-hop ad-hoc networks. They also arrive at the logarithmic scaling of throughput with the number of users for the Rayleigh fading case. Moreover, they derive the throughput for networks with Log-normal and Weibull fading. Although \cite{Gesbert}, \cite{Sohn}, and \cite{Kountouris}  address a problem similar to ours, it should be noted that there are major differences between our work and those papers. First, the logarithmic scaling of the throughput in aforementioned references is due to use of path-loss effect in the model. In our paper, we do not consider path loss, resulting in a critically interference-limited scenario, which in many scenarios is a more proper model for large ad-hoc networks deployed in a limited area. Second -- in contrast to our paper -- they consider a down-link scenario. Thus, the base station is allowed to schedule users and can choose the best user for each channel realization. However, in our model, the source-destination pairs are specified and remain fixed during network operation. Consequently, the aforementioned references rely on extreme order statistics results, while we use the results from intermediate order statistics. The paper \cite{Kountouris} also considers a two-hop network, which provides the flexibility of choosing the appropriate relay for each source-destination pair. So, the results of \cite{Kountouris} are more comparable with \cite{Cui_Opportunistic},  as they both consider a two-hop network, where our model is one-hop.

Finally, another work that also relies on down-link scheduling is presented in \cite{Aggarwal}. Considering the down-link scenario leads to a multiuser diversity gain which is not present in our model. In that sense, the results in \cite{Aggarwal} should be compared with the two-hop networks as in \cite{Cui_Opportunistic}, and \cite{Kountouris}, where such flexibility is provided by the use of relay nodes.

In summary, the major contribution of our paper is deriving throughput scaling for a completely general power distribution in closed form. This achievement is due to our choice of a totally different approach where we analyze the problem in the context of \emph{intermediate order statistics}. As a result, the results of other papers for the same model can be reproduced as special cases of our main theorem (as we will illustrate in corollaries throughout the paper). In addition, it is important to note that our results provide the opportunity to derive throughput scaling for other new distributions, which have not been addressed in earlier work (e.g., Weibull distribution in Corollary \ref{Cor_Lower_Weibull}, or Nakagami-$m$ distribution in Corollary \ref{Cor_Lower_Nakagami}). Finally, we propose an upper bound for the large class of super-exponential distributions (the upper bound proposed in previous results only covers the Rayleigh fading case.). Such approach paves the path for proving the order-optimality of the heuristic scheme for Nakagami-$m$ fading (extending the previous result for Rayleigh fading).

\subsection{Notations and Paper Organization}

We use Knuth's asymptotic notation as \cite{Cormen}: $f(n)=\mbox{O}(g(n))$ if there exist positive constants $\tau$ and $n_0$ such that for all $n\geq n_0$ we have $0\leq f(n) \leq \tau g(n)$, $f(n)=\Omega(g(n))$ if $g(n)=\mbox{O}(f(n))$,  $f(n)=\Theta(g(n))$ if $f(n)=\mbox{O}(g(n))$ and $f(n)=\Omega(g(n))$. The expression ``$f_1(x) \rightarrow f_2({x}) \hspace{1em} \rm as \it \hspace{1em} x \rightarrow \infty$'' is equivalent to $\lim_{x\rightarrow\infty}\frac{f_1(x)}{f_2(x)}=1$. In addition, $c_i$'s ($i=1,\dots,7$) are strictly positive constants independent of $n$, where $n$ is the number of the source-destination pairs in the network. Also, an event $\mathbb{B}$ in the network happens with high probability (w.h.p.), if and only if, $\lim_{n \rightarrow \infty}\Pr\{\mathbb{B}\}=1$. Finally, all the logarithms are to base $e$, unless stated otherwise.

The rest of the paper is organized as follows. In Section \ref{Sec_Network_Model}, we describe the network model and explain our problem. In Section \ref{Sec_Lower} we state the main lower bound theorem and discuss the proof sketch. Then, in Section \ref{Sec_Cases}, the theorem is applied to four practical channel power distributions as case studies. In Section \ref{Sec_Upper}, we propose an upper bound for throughput of networks with super-exponential distribution, and discuss its implications. Finally, Section \ref{Sec_Conclusion} concludes the paper.

\section{Network Model and Problem Definition}\label{Sec_Network_Model}

In this section, we first explain the network model. Then, we describe the problem we consider under such network model. Finally, we discuss the computational complexity of the problem as the number of nodes in the network grows.

Consider $n$ wireless nodes transmitting their data to $n$ wireless receive nodes, in a shared medium. The transmit nodes are denoted by $S_1,\dots,S_n$, and the receive nodes are denoted by $D_1,\dots,D_n$. The node $S_i$ sends its data to the node $D_i$, and $D_i$ is only interested in the data sent by $S_i$. The communication is assumed to be completed in a single hop where it follows an on-off paradigm. In other words, at each time slot, a subset of transmit nodes (i.e. $\mathbb{S} \subset \{S_1,...,S_n\}$) are active and send their signal with maximum power $P$, while the remaining transmitters are inactive and do not transmit. Without loss of generality, we consider $P=1$ throughout the paper.

At each time slot, the channel power gain between transmitter $S_i$ and receiver $D_j$ is modeled by a random variable $\gamma_{i,j}$. The random variables $\gamma_{i,j}$ are assumed to be independent and identically distributed (i.i.d.). In other words, $\gamma_{i,j}$'s are drawn from the common probability distribution function (p.d.f.) $f(\gamma)$ in an independent manner. Also, we define $\mathbb{E}\{\gamma_{i,j}\}=\mu$, where $\mathbb{E}\{.\}$ represents the expectation operator. The channel power coefficients follow a quasi-static rule. In other words, during a single time slot the channel power gains are fixed. However, at the next time slot they are changed independently from other time slots. Fig. \ref{Fig_Schematic} shows the network model. We stress the fact that this i.i.d. model for the link power between the nodes is a widely-accepted model in the literature, and it has been used in many papers such as \cite{Ebrahimi}, \cite{Ebrahimi_2011}, \cite{Gowikar}, \cite{Cui_Opportunistic}, and \cite{Cui} .

We assume single-user decoding at each receive node. We define the Signal to Interference and Noise Ratio (SINR) at the receive node $D_i$ to be
\begin{equation}\label{Eq_SystemModel_SINR_Def}
	SINR_{i}\triangleq \frac{\gamma_{i,i}}{N_0+\sum_{S_k \in \mathbb{S}, k\neq i}{\gamma_{k,i}}},
\end{equation}
where $N_0$ is the variance of Additive White Gaussian Noise (AWGN) at the receivers. Thus, for the receive node $D_i$ to be successful at a specific time slot, we should have $SINR_{i}>\beta$ at that time slot, where $\beta$ is a constant. To make the presentation more readable we define the following terms for our network model:

\begin{mydef}\label{Def_Direct_Cross_links}
The links from $S_i$ to $D_i$, $i=1,\dots,n$ are called \emph{Direct Links}, while the links from $S_i$ to $D_j$, $i,j=1,\dots,n$, $i \neq j$, are called \emph{Cross Links}, or \emph{Interference Links}.
\end{mydef}

\begin{mydef}\label{Def_Activation_Vector}
The node activation vector, denoted by $\mathbf{x}$, is an $n \times 1$ binary vector where the $i$th element of $\mathbf{x}$ is equal to $1$, if $S_i$ is active and transmits with full power, and is $0$ otherwise.
\end{mydef}

\begin{mydef}\label{Def_Throughput}
The network throughput, denoted by $T$, is defined as the number of successful receive nodes (i.e. the ones that satisfy the $SINR$ constraint). It should be noted that $T$ depends on the active set $\mathbb{S}$.
\end{mydef}

In our model, we assume that the information of all the channel states (i.e. $\gamma_{i,j}$, $i,j=1,\dots,n$) is available at all nodes in the network.

With the above model and assumptions, we will consider the following problem: At a specific time slot, by knowing the channel power gains (i.e. $\gamma_{i,j}$'s) our goal is to find the optimum subset of active nodes such that the largest number of successful receptions at the receivers is achieved. In other words, we should address the following optimization problem
\begin{equation}\label{Eq_Model_Optimization_Problem}
    \mathbb{S}^*={\arg{\max}}_{\mathbb{S} \subset \{S_1,...,S_n\}}   T(\mathbb{S}),
\end{equation}
where the corresponding throughput will be
\begin{equation}\label{Eq_Model_Optimum_Throughput}
    T^*=T(\mathbb{S}^*).
\end{equation}

Such optimization problem is inherently complex as success or failure of each source is tightly coupled with the status of other sources.
For example, consider a source which has a strong direct link. Such source will have a high chance of success if it is activated, however, it can deteriorate the chance of other nodes if it creates strong interference links towards them. Therefore, in choosing the subset of active nodes we should consider all the information we have about direct and cross links, in order to achieve the maximum number of successful receptions.

The most trivial algorithm for solving (\ref{Eq_Model_Optimization_Problem}) is an exhaustive search over all subsets of $\{S_1,...,S_n\}$, and finding the one resulting in the largest throughput. The complexity of such algorithm is of order $2^n$, which makes it impractical. Thus, we look for more effective algorithms to find optimal or near-optimal solutions to (\ref{Eq_Model_Optimization_Problem}).

In order to get a better understanding of the problem complexity, we develop the following formulation based on the notion of activation vector. Accordingly, we will have
\begin{equation}\label{Eq_Model_SINR_Based_on_Activation_Vector}
	SINR_{i} = \frac{x_i\gamma_{i,i}}{N_0+\sum_{k\neq i}{x_k\gamma_{k,i}}},
\end{equation}
for $i=1,\dots,n$. Consequently, the successful reception condition at the receive node $D_i$ (i.e. $SINR_{i}>\beta$) will be
\begin{equation}\label{Eq_Model_SINR_Success_Based_on_Activation_Vector}
	\gamma_{i,i}x_i-\beta\sum_{k\neq i}{\gamma_{k,i}x_k} > \beta N_0.
\end{equation}
Define the $n \times n$ matrix
\begin{equation}\label{Eq_Model_MaxFS_Matrix_Definition}
\mathbf{A} \triangleq \left( \begin{array}{cccc} \gamma_{1,1} & -\beta\gamma_{2,1} & \dots & -\beta\gamma_{n,1} \\
                        -\beta\gamma_{1,2} & \gamma_{2,2}, & \dots & -\beta\gamma_{n,2}  \\
                        \vdots & \vdots & & \vdots\\
                         -\beta\gamma_{1,n} & -\beta\gamma_{2,n} & \dots & \gamma_{n,n}   \end{array} \right),
\end{equation}
and the $n \times 1$ vector
\begin{equation}\label{Eq_Model_MaxFS_Vector_Definition}
\mathbf{b} \triangleq \left( \begin{array}{c} \beta N_0 \\
                        \beta N_0  \\
                        \vdots \\
                         \beta N_0   \end{array} \right).
\end{equation}
Then, the set of successful reception conditions at all receive nodes (the set of inequalities in (\ref{Eq_Model_SINR_Success_Based_on_Activation_Vector})) can be formulated by the following set of linear inequalities
\begin{equation}\label{Eq_Model_MaxFS_Set_of_Inequalities}
    \mathbf{A}\mathbf{x} > \mathbf{b},
\end{equation}
where $\mathbf{x}$ is the activation vector. For a specific choice of the activation vector, the number of inequalities satisfied in (\ref{Eq_Model_MaxFS_Set_of_Inequalities}) is exactly equal to the number of receive nodes which satisfy the $SINR$ constraint, i.e. the network throughput. Thus, solving the optimization problem in (\ref{Eq_Model_Optimization_Problem}) is equivalent to finding the binary activation vector $\mathbf{x}^*$ resulting in the largest number of inequalities satisfied in (\ref{Eq_Model_MaxFS_Set_of_Inequalities}).

This problem is called the \emph{maximum feasible subsystem problem} (Max FS problem) and arises in many other research fields such as machine learning, political science, computational biology, and $\dots$ \cite{MaxFS1}, \cite{MaxFS2}. This problem (in the case of binary vector $\mathbf{x}$, which is the case considered in this paper) is shown to be at least as hard as finding a maximal independent set in a graph\footnote{By definition, an independent set in a graph is a set of vertices, where no two of them are adjacent. Accordingly, a maximal independent set in a graph is an independent set which is not a subset of any other independent set.}, which is NP-hard \cite{MaxFS1}. A number of efficient sub-optimum algorithms have been developed to address this problem \cite{MaxFS2}, \cite{MaxFS3}. The essence of all these algorithms is proposing a numerical approach with no rigorous performance guarantee. Therefore, we cannot exploit them to get closed-form results for the network throughput.

\section{Lower Bound on the Throughput}\label{Sec_Lower}

In this section, we first propose a simple heuristic method to design the activation vector, and then analyze the resulting throughput. The main idea of the underlying scheme is to activate the source-destination pairs with the best direct links, and let the remaining pairs be silent. Thus, in this method, we do not use the information regarding the power of interference links in order to decide which source nodes to activate. Therefore, the complexity of this scheme is polynomial in terms of the number of nodes. Theorem \ref{Th_Main_Theorem_Lower_Bound} is the main result characterizing the throughput of the network operated with this strategy which holds for a class of channel power distributions with a given number of properties. Before stating the theorem, we need to define these properties:

\begin{mydef}\label{Def_Condition_Set1}
	The random variable $X$ is said to satisfy the condition set 1 if its p.d.f., $f(x)$, and its cumulative distribution function (c.d.f.), $F(x)$, satisfy the following conditions:
\begin{equation}\label{Eq_Condition_Set1_1}
	F^{-1}(1) \rightarrow \infty,
\end{equation}
\begin{equation}\label{Eq_Condition_Set1_2}
 \lim_{x \rightarrow \infty} x h(x) =c_0 >0, \hspace{1em} \rm or \it \hspace{1em} \lim_{x \rightarrow \infty} \frac{\textrm{d}}{\textrm{d} x}\frac{1}{h(x)}=\rm 0 ,
\end{equation}
 where $F^{-1}(x)$ represents the inverse of the function $F(x)$, $h(x) \triangleq \frac{f(x)}{1-F(x)}$, and $c_0$ is a constant.
\end{mydef}

\begin{mydef}\label{Def_Condition_NHT_RV}
	The random variable $X$ is said to be of the Non-Heavy-Tailed (NHT), or super-exponential, type if $\mathbb{E}\{e^{t X}\}<\infty$, for some $t>0$ (Cram\'{e}r's condition).
\end{mydef}

\begin{mydef}\label{Def_Condition_HT_RV}
	The random variable $X$ is said to be of the Heavy-Tailed (HT), or sub-exponential, type if it has one of the following properties\footnote{It should be noted that this is not an exhaustive list of heavy-tailed distributions. We only mention the ones for which our theorem applies.}:
\begin{itemize}

\item
Regularly varying tail:

\begin{equation}\label{Eq_Def_HT_RV_Regularly_Varying_Tail}
    1-F(x)=\frac{L(x)}{x^\alpha}, \hspace{3pt} \textrm{for} \hspace{3pt} x>0
\end{equation}
for $\alpha>2$, and where, $L(x)$ is a slowly-varying function.

\item
Log-normal type tail:
\begin{equation}\label{Eq_Def_HT_RV_LogNormal_Tail}
    1-F(x)\sim cx^\beta e^{-\lambda\log^\gamma x}, \hspace{3pt} \textrm{as} \hspace{3pt} x\rightarrow \infty
\end{equation}
for $\gamma>1$ and $\lambda>0$.

\item
Weibull-like tail:
\begin{equation}\label{Eq_Def_HT_RV_Weibull_Like_Tail}
    1-F(x) \sim cx^\beta e^{-\lambda x^\alpha}, \hspace{3pt} \textrm{as} \hspace{3pt} x\rightarrow \infty
\end{equation}
for $0<\alpha<0.5$ and $\lambda>0$.
\end{itemize}

For a rigorous definition of heavy-tailed distributions refer to \cite{Feller}.

\end{mydef}

Now, we are ready to state the main theorem:
\begin{thm}\label{Th_Main_Theorem_Lower_Bound}
If the power distribution of the underlying wireless channel (with the cumulative distribution function (c.d.f.) $F(x)$) satisfies the condition set 1, and is of NHT or HT type (as defined in definitions \ref{Def_Condition_NHT_RV} and \ref{Def_Condition_HT_RV}, respectively), then in a one-hop wireless network with $n$ source-destination pairs, throughput is lower bounded (w.h.p.) by
\begin{equation}\label{Eq_Main_Theorem_Lower_Bound_General_Throughput}
 T=\Omega\left( G^{-1}\left(n\right)\right) ,
\end{equation}
where,
\begin{equation}\label{Eq_Main_Theorem_Lower_Bound_G_X}
	G(x) \triangleq \frac{x}{1-F\left(\beta\mu x/2\right)},
\end{equation}
in which, $\beta$ is the $SINR$ target, and $\mu$ is the average channel power.
\end{thm}
\begin{proof}
See Appendix A.
\end{proof}

First, one should note that the distributions to which Theorem \ref{Th_Main_Theorem_Lower_Bound} applies are completely general, since it covers both the super-exponential and sub-exponential distributions. Also, condition set 1 is completely general and is satisfied by most practical distributions \cite{Arnold}. Later (in corollaries \ref{Cor_Lower_Nakagami} to \ref{Cor_Lower_Log_Normal})  we will provide several examples of how to apply this theorem to different distributions.

Next, we describe the main ideas behind the proof of this theorem. Theorem \ref{Th_Main_Theorem_Lower_Bound} states that there exists a method according to which we can have order of $G^{-1}(n)$ concurrent successful transmissions in the network. In other words, it states that we can find an activation vector which results in the order of $G^{-1}(n)$ satisfied inequalities in (\ref{Eq_Model_MaxFS_Set_of_Inequalities}). The underlying method relies on activating source-destination pairs with the best direct links. Suppose we activate $t_1$ of them in the following way:

First, sort $\gamma_{i,i}$'s to get the order statistics $\gamma_{(i),(i)}$'s such that\footnote{In other words $\gamma_{(i),(i)}$'s are the sorted version of  $\gamma_{i,i}$'s.}
\begin{equation}\label{Eq_Th1_Proof_Idea_Sorting}
	\gamma_{(1),(1)} \leq \gamma_{(2),(2)} \leq  \dots \leq \gamma_{(n),(n)}.
\end{equation}
Accordingly, we can sort the corresponding source-destination pairs $S_i$-$D_i$ as $S_{(i)}$-$D_{(i)}$. In other words, the pair $S_{(i)}$-$D_{(i)}$ has a better or the same quality direct link comparing with $S_{(j)}$-$D_{(j)}$ if $i > j$. Then, the proposed candidate subset for the set of active nodes will be:
\begin{equation}\label{Eq_Th1_Proof_Idea_Candidate_Set}
	\mathbb{S}= \{ S_{(i)} |  i=n-t_1+1, \dots, n\}.
\end{equation}

This idea is illustrated in Fig. \ref{Fig_Proof_Sketch} where the source-destination pairs are sorted according to their direct links, and the first $t_1$ strongest ones are activated. Then, the main point of Theorem \ref{Th_Main_Theorem_Lower_Bound} is that we can choose $t_1$ as large as $G^{-1}(n)$ ($G(x)$ is defined in (\ref{Eq_Main_Theorem_Lower_Bound_G_X})), where all receptions will be successful. In other words, if we activate $G^{-1}(n)$ source-destination pairs with the best direct links, and force the remaining pairs to be inactive, all active pairs will satisfy the $SINR$ constraint.

Next, we describe why all these receptions satisfy the $SINR$ constraint. First, we note that the most critical transmission among active pairs is the one from $S_{(n-t_1+1)}$ to $D_{(n-t_1+1)}$. If such transmission is successful, then all other transmissions will be successful as well (see Fig. \ref{Fig_Proof_Sketch}). That is due to the fact that all active pairs experience statistically equivalent interference, while, this pair has the weakest direct link among all of them. Therefore, we should choose $t_1$ as large as to ensure that this specific transmission will be successful. If we enlarge $t_1$ (i.e. the number of active nodes) it means we are choosing a larger set of active transmitters. This, in return, means that the direct channel power of the weakest link among active pairs (i.e. $\gamma_{(n-t_1+1),(n-t_1+1)}$) decreases. Also, enlarging $t_1$ will result in an increase in the interference level of receivers. One can see that the power of the weakest direct link will be of order $F^{-1}(1-t_1/n)$, and the increase of the interference is linear with $t_1$ (see Fig. \ref{Fig_Proof_Sketch}). If power of direct link of this transmission pair is greater than the power of the interference imposed on this pair by other transmissions, then this reception will be successful. Thus, the maximum $t_1$ that we are able to choose occurs when the order of these two quantities coincide. Roughly speaking, it can be examined that such situation happens when we activate the order of $G^{-1}(n)$ nodes. Therefore, by activating the order of $G^{-1}(n)$ sources, we still have non-vanishing $SINR$ at all corresponding receivers, while all of them satisfy the $SINR$ constraint (as will be shown rigorously in the proof, in Appendix A).

The fact that $\gamma_{(n-t_1+1),(n-t_1+1)}$ is of order $F^{-1}(1-t_1/n)$ comes from a result in the ``Intermediate Order Statistics'' context which is explained in Lemma \ref{Lemma_Lower_Falk} in Appendix A. Also, the fact that the interference increases linearly with the number of active nodes can be made precise by the help of results from the ``Large Deviations Principle'' context, which is explained in Lemmas \ref{Lemma_Lower_LDP_NHT} and \ref{Lemma_Lower_LDP_HT} in Appendix A. The rigorous proof of Theorem \ref{Th_Main_Theorem_Lower_Bound} along with mathematical techniques used are provided in Appendix A.

\begin{figure}
\begin{center}
\includegraphics[width=0.7\textwidth]{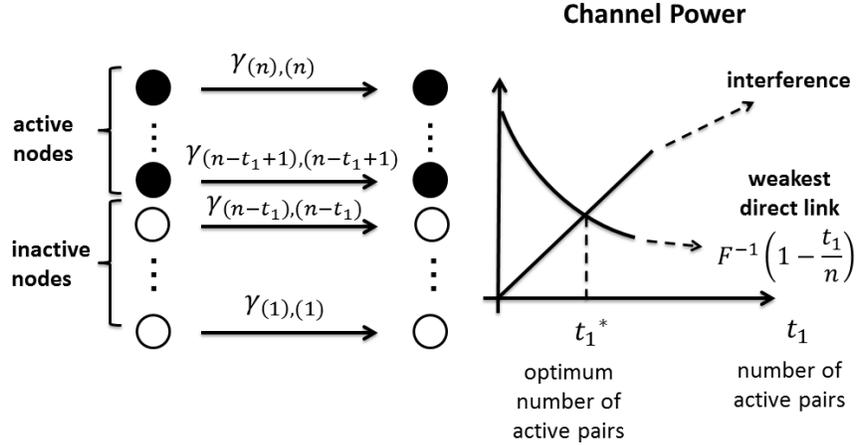}
\end{center}
\caption{The proposed scheme and the proof sketch.\label{Fig_Proof_Sketch}}
\end{figure}

\section{Case Studies}\label{Sec_Cases}

Theorem \ref{Th_Main_Theorem_Lower_Bound} holds for the class of distributions satisfying the condition set 1, which are of NHT or HT type (as defined in definitions \ref{Def_Condition_NHT_RV} and \ref{Def_Condition_HT_RV}, respectively). As mentioned earlier, these conditions are very general, and accordingly, the theorem applies to a wide range of distributions. Next, we present four corollaries as examples of applying this theorem to specific distributions. In order to apply Theorem \ref{Th_Main_Theorem_Lower_Bound} to a specific channel power distribution, we should follow three steps. First, we have to check if that distribution satisfies the conditions of theorem. Second, we should find the order of growth of $G^{-1}(x)$ as $x$ gets large. Finally, from the theorem we know that achievable throughput is of order $G^{-1}(n)$  (rigorous proofs of corollaries can be found in Appendix B).

\begin{cor}[Achievable Throughput of Networks with Nakagami-$m$ Fading (Gamma Power Distribution)]\label{Cor_Lower_Nakagami}
If in an environment the channel gain behaves under the Nakagami-$m$ distribution, then the channel power follows the Gamma distribution whose p.d.f. is given by (all the distribution functions in this paper represent power of the channel, and thus, are defined for just non-negative values)
\begin{equation}\label{Eq_Cor_Nakagami_PDF}
	f(x)=x^{m-1}\frac{e^{-mx/\Omega}}{\Gamma(m)\left(\frac{\Omega}{m}\right)^m},
\end{equation}
and the corresponding c.d.f. is as follows
\begin{equation}\label{Eq_Cor_Nakagami_CDF}
	F(x)=\frac{\gamma \left( m,\frac{mx}{\Omega} \right) } {\Gamma(m)},
\end{equation}
where $\Gamma(.)$ represents the Gamma function, $\gamma(.,.)$ is the lower incomplete Gamma function,  $m$ and $\Omega$ are parameters of the distribution.
In such network, the achievable throughput will be
\begin{eqnarray}\label{Eq_Cor_Nakagami_Throughput}
	T=\Omega\left(\log(n)\right).
\end{eqnarray}

\end{cor}

\begin{cor}[Achievable Throughput of Networks with Weibull Power Distribution]\label{Cor_Lower_Weibull}

A random variable is said to have Weibull distribution if its p.d.f. and c.d.f. are as follows, respectively:
\begin{eqnarray}\label{Eq_Cor_Weibull_PDF}
	f(x) =\frac{k}{\lambda}\left(\frac{x}{\lambda}\right)^{k-1}e^{-\left(x/\lambda\right)^k},
\end{eqnarray}
\begin{equation}\label{Eq_Cor_Weibull_CDF}
F(x)=1-e^{-(x/\lambda)^k},
\end{equation}
where $k$ and $\lambda$ are parameters of the distribution. Then, the throughput of the network will be
\begin{equation}\label{Eq_Cor_Weibull_Throughput}
 T=\Omega\left(\left(\log(n)\right)^{1/k}\right).
\end{equation}
\end{cor}

\begin{cor}[Achievable Throughput of Networks with Generalized Pareto Power Distribution]\label{Cor_Lower_Pareto}
Consider the following channel power distribution
\begin{equation}\label{Eq_Cor_Pareto_PDF}
    f(x)=\frac{\alpha}{(1+x)^{\alpha+1}},
\end{equation}
\begin{equation}\label{Eq_Cor_Pareto_CDF}
    F(x)=1-\frac{1}{(1+x)^\alpha},
\end{equation}
where $\alpha>2$. Then, the achievable throughput in this case will be
\begin{eqnarray}\label{Eq_Cor_Pareto_Throughput}
	T=\Omega\left(n^{1/(1+\alpha)}\right).
\end{eqnarray}

\end{cor}

\begin{cor}[Achievable Throughput of Networks with Log-normal Power Distribution]\label{Cor_Lower_Log_Normal}
Consider the following channel power distribution
\begin{equation}\label{Eq_Cor_LogNormal_PDF}
    f(x)=\frac{1}{\sqrt{2\pi \sigma^2}}\frac{1}{x}e^{\frac{-(\log{x}-\mu)^2}{2\sigma ^ 2}},
\end{equation}
\begin{equation}\label{Eq_Cor_LogNormal_CDF}
    F(x)=\frac{1}{2}+\frac{1}{2} \rm{erf} (\frac{\log{\it x}-\mu}{\sqrt{2\sigma ^ 2}}),
\end{equation}
where $\sigma$ and $\mu$ are the parameters of the distribution. Then, the achievable throughput in this case will be
\begin{eqnarray}\label{Eq_Cor_LogNormal_Throughput}
	T=\Omega\left(e^{\sqrt{2}\sigma\sqrt{\log(n)}}\right).
\end{eqnarray}

\end{cor}

It is important to note a few issues that can be deduced from the corollaries. First, considering Corollary \ref{Cor_Lower_Pareto}, the throughput for the Pareto distribution depends on the parameter $\alpha$, which determines how fast $1-F(x)$ decays, as $x$ gets large. In fact, for smaller values of $\alpha$ we have a heavier tail. Similarly, from (\ref{Eq_Cor_Pareto_Throughput}) we see that for smaller values of $\alpha$, higher throughput can be achieved, which leads to have higher throughput for heavier tails. The same trend is noted for the Weibull distribution in Corollary \ref{Cor_Lower_Weibull}. By increasing the parameter $k$, we will suppress the tail of $1-F(x)$, while, the throughput will also decrease. We can explain this phenomenon with the help of concepts introduced through the proof sketch where we demonstrated that the most important factor determining the throughput for a distribution is the amount of multiuser diversity gain it can provide. In other words, the distributions that have a heavier tail will result in a larger value for $t_1^*$ (this is illustrated in Fig \ref{Fig_Proof_Sketch}). That is due to the fact that in distributions with heavier tails, we have a higher chance of encountering high power channels, which will provide the chance to benefit from them in the multiuser diversity context\footnote{It should be noted we are interested in the intermediate order statistic behavior of the underlying distribution, not the extreme order statistics. However, our intuition says that a distribution with high extreme order statistics will naturally yield a high intermediate order statistics as well.}. If we constrain the distribution to have a finite channel power, then the heaviest tail will belong to the Pareto distribution with $\alpha=2+\epsilon$ (for any small strictly positive $\epsilon$)\footnote{This is in agreement with the result in \cite{Pooya}.}. In the Log-normal distribution, the parameter $\sigma$ determines the shaping of tail of the distribution, and thus, determining the multiuser diversity gain. In the Nakagami-$m$ fading, although $m$ determines how heavy the tail is, it is not effective-enough to increase (or decrease) the multiuser diversity gain order. Thus, the throughput of Nakagami-$m$ fading is of order $\log(n)$, regardless of value of $m$.

The above corollaries characterize the throughput scaling of wireless networks with different link power distributions\footnote{Specifically, Corollary \ref{Cor_Lower_Nakagami}, when considered in the special case of $m=1$, covers the results in \cite{Ebrahimi} and \cite{Ebrahimi_2011}, Corollary \ref{Cor_Lower_Log_Normal} covers the result in \cite{Ebrahimi_report}, and Corollary \ref{Cor_Lower_Pareto} covers the result in \cite{Pooya} in a stronger probabilistic sense. The result in Corollary \ref{Cor_Lower_Weibull} for Weibull distribution has not appeared in previous papers.}, based on the scheme provided in Theorem \ref{Th_Main_Theorem_Lower_Bound}. An important question still to be addressed is as follows. In the scheme provided by Theorem \ref{Th_Main_Theorem_Lower_Bound} to achieve the throughput stated in (\ref{Eq_Main_Theorem_Lower_Bound_General_Throughput}), we have just used the information about the status of direct links. However, in the optimum scheme, one should exploit the information regarding the status of cross links as well. Thus, we have to characterize how much throughput is lost due to this simplification in the scheme. In the next section, we will address this issue by providing throughput upper bound for all possible activation vectors, and will compare it with the current achievable throughput.

\section{Upper Bound on the Throughput }\label{Sec_Upper}
In the previous section, we have established lower bound results for the throughput of the network. In this section, we present an upper bound on the throughput in the case of NHT type (super-exponential) distributions in Theorem \ref{Th_Upper_NHT}.

\begin{thm}\label{Th_Upper_NHT}
	Throughput of one-hop schemes in wireless networks with link power distribution of NHT (super-exponential) type  is upper bounded by $\mbox{O}\left(\log(n)\right)$.
\end{thm}

\begin{proof}[Proof of Theorem \ref{Th_Upper_NHT}]
See Appendix C.
\end{proof}

According to Theorem \ref{Th_Upper_NHT}, there exists no strategy for activating the source nodes which results in a throughput of order more than $\log(n)$.

We can now present the following corollary for the throughput scaling of networks with Nakagami-$m$ fading:

\begin{cor}\label{Th_Nakagami}
	Throughput of one-hop schemes in wireless networks with Gamma link power distribution (Nakagami-$m$ fading) is of order $\Theta\left(\log(n)\right)$.
\end{cor}

\begin{proof}[Proof of Corollary \ref{Th_Nakagami}]
Gamma is a super-exponential distribution. According to Theorem \ref{Th_Upper_NHT}, the throughput of all activation strategies is upper bounded by order of $\log(n)$. On the other hand, we have shown in Corollary \ref{Cor_Lower_Nakagami} that there exists an activation strategy that achieves the throughput of the same order for networks with Nakagami-$m$ fading. Since the upper and lower bounds meet, we conclude that the throughput is of order $\Theta\left(\log(n)\right)$.
\end{proof}

Thus, we conclude that activating the pairs with strongest direct links -- which ignores information of interference links -- is order-optimal for Nakagami fading. Meanwhile, we observe that for the Weibull distribution, upper bound of throughput is of order $\log(n)$ according to Theorem \ref{Th_Upper_NHT}. However, according to Corollary \ref{Cor_Lower_Weibull} the lower bound we have derived in this paper for Weibull distribution is of order $\left(\log(n)\right)^ {1/k}$. This means that we do not yet know whether the scheme (used in the proof of Theorem \ref{Th_Main_Theorem_Lower_Bound}) is order-optimal for the Weibull distribution or not.

Also, Theorem \ref{Th_Upper_NHT} has an interesting algorithmic implication for super-exponential distributions. Generally, in order to find the optimum activation vector for a generic distribution one needs to do $2^n$ search trials on all activation vectors. However, for super-exponential distributions, we know that no more than $\log(n)$ nodes can be active in the optimum activation vector. This means that we should search over activation vectors whose weight (by definition, weight of a binary vector is the number of $1$'s in that vector) is less than $\log_2(n)$. Thus, in this case we need the following number of search steps\footnote{Changing the logarithm base from $e$ to $2$ is just for presentation simplicity, and does not has any effect on scaling.}:
\begin{eqnarray}\label{Eq_Upper_Algorithmic}
s&=&\sum_{i=1}^{\log_2(n)}{{{n}\choose{i}}} \\ \nonumber
&<&\log_2(n){{n}\choose{\log_2(n)}} \\ \nonumber
&<&\log_2(n) n^{\log_2(n)} \\ \nonumber
&=&\log_2(n) 2^{\log_2^2(n)}
\end{eqnarray}
which is significantly less than the original $2^n$ number of search trials for a generic distribution. Also, finding upper bounds for other distributions (other than super-exponential) will be of the same algorithmic importance as well.

\section{Conclusion and Future Work}\label{Sec_Conclusion}

In order to understand fundamental limits of wireless networks and design corresponding transmission strategies, we have proposed lower and upper bounds for network throughput of one-hop networks. Our main lower bound theorem proposes a closed-form throughput expression for a network with general fading. In addition to reducing to previous results (Rayleigh fading, Log-normal power distribution, and Pareto distribution, in \cite{Ebrahimi}, \cite{Ebrahimi_report}, and \cite{Pooya}, respectively) as its special cases, this theorem enables us to also derive throughput for new distributions (Nakagami-$m$ fading and Weibull power distributions). More importantly, our analysis approach explains the mechanism behind the heuristic method, and the main phenomenon affecting its performance. In fact, the multiuser diversity gain, which we have characterized in the context of intermediate order statistics, determines the throughput of the network; the heavier the power distribution tail is, the higher the throughput will be.

Moreover, we have proposed a throughput upper bound of order $\log(n)$ for super-exponential power distributions. Based on this upper bound, we have proved that the heuristic scheme is order-optimal for Nakagami-$m$ fading. It means that, to decide about the transmission strategy we do not need to use information regarding cross links, for networks with Nakagami-$m$ fading. However, the throughput upper bound of order $\log(n)$ does not coincide with the achievable throughput of order $\log^{1/k}(n)$ for the case of Weibull distribution. This leaves the problem of finding order optimal schemes for networks with Weibull channel distribution open. In addition, there is still need for proposing throughput upper bounds for sub-exponential distributions (e.g., Pareto and Log-normal distribution) to better understand the behavior of heuristic methods in such cases.

Two important practical issues remain untouched, which are interesting topics for future works:

\subsection{Channel State Information}
One important practical issue regarding the scheme achieving the mentioned throughput in Theorem \ref{Th_Main_Theorem_Lower_Bound} is the amount of the Channel State Information (CSI) it requires. In the scheme used in Theorem \ref{Th_Main_Theorem_Lower_Bound}, each destination should decide whether or not its direct link with the corresponding source is among the $t_1$ best direct links. Then, for a positive answer, the destination instructs the corresponding source to become active at that time slot. Thus, the destination should know the power of the direct link for all transmission pairs (i.e. $\gamma_{i,i}, i=1,\dots,n$) in order to perform the sorting process. Such approach requires some means  of communication between destinations in order to share CSI. As an alternative method, if the direct link power of each destination is above a carefully-designed threshold level, then the corresponding source should become active (similar to the idea used in \cite{Ebrahimi}). Therefore, in this alternative approach, destinations do not need to know the power of direct links of other transmissions, relaxing the need for inter-destination CSI exchange. A rigorous study of this idea, or other schemes relaxing the need for substantial amount of CSI, is an interesting topic for future work.

\subsection{Delay and Fairness}
Since the method used in Theorem \ref{Th_Main_Theorem_Lower_Bound} is benefiting from the multiuser diversity gain, some price in terms of delay should be paid. If we define the delay to be the number of time slots it takes for a node to become active, then it is easy to note that the average delay for any scheme with the achievable throughput of order $T(n)$ will be of order $n/T(n)$ for each node\footnote{It is important to note that this is not a throughput-delay trade-off characterization. This is due to the fact that throughput and delay results are derived for the case of maximum throughput without taking delay performance into account.}. Also, it should be noted that since the nodes are statistically identical in terms of channel powers, each one will have the same share of total throughput in an extended time span, ensuring the fairness of the scheme. A rigorous study of probabilistic behavior of the delay of heuristic scheme, and considering fairness in a limited time window are other interesting topics for further research.

\newpage

\section*{Appendix A:  Proof of Theorem \ref{Th_Main_Theorem_Lower_Bound}}\label{App1}

Before proving the theorem we need some lemmas. The first lemma comes from the ``Order Statistics'' context and is an assertion about ``Intermediate Order Statistics''.

\begin{lem} [Falk, 1989] \label{Lemma_Lower_Falk}
Assume that $X_1,X_2,\dots,X_n$ are i.i.d. random variables with the p.d.f. $f(x)$. Assume that $F(x)$ (the corresponding c.d.f.) is absolutely continuous and satisfies the condition set 1. Define $X_{(1)},X_{(2)}, \dots, X_{(n)}$ to be the order statistics of $X_1,X_2,\dots,X_n$. If $i \rightarrow \infty$ and $i/n \rightarrow 0$ as $n \rightarrow \infty$, then there exist sequences $a_n$ and $b_n>0$ such that
\begin{equation}\label{Eq_Lemma_Falk_1}
	\frac{X_{(n-i+1)}-a_n}{b_n} \Rightarrow N(0,1),
\end{equation}
where $\Rightarrow$ denotes convergence in distribution, and $N(0,1)$ is the Normal distribution with zero mean and unit variance. Furthermore, one choice for $a_n$ and $b_n$ is:
\begin{eqnarray}\label{Eq_Lemma_Falk_2}
	a_n=F^{-1}(1-\frac{i}{n}), \;\;\;	b_n=\frac{\sqrt{i}}{nf(a_n)}.
\end{eqnarray}
\end{lem}
\begin{proof}[Proof of Lemma \ref{Lemma_Lower_Falk}]
The proof of Lemma \ref{Lemma_Lower_Falk} can be found in \cite{Falk} (see also \cite{Arnold}, Ch. 8).
\end{proof}

In summary, Lemma \ref{Lemma_Lower_Falk} states that the random variable $X_{(n-i+1)}$ is asymptotically a standard Normal random variable, after being normalized by the sequences $a_n$ and $b_n$. Also, we need the following Lemma which is closely related to the previous one:

\begin{lem}\label{Lemma_Lower_an_bn}
	In Lemma \ref{Lemma_Lower_Falk} we have
	\begin{equation}\label{Eq_Lemma_an_bn}
		\lim_{n \rightarrow \infty} \frac{a_n}{b_n}=\infty,
	\end{equation}
    where $a_n$ and $b_n$ are defined in (\ref{Eq_Lemma_Falk_2}).
\end{lem}

\begin{proof}[Proof of Lemma \ref{Lemma_Lower_an_bn}]
	\begin{eqnarray}\label{Eq_Proof_Lemma_an_bn_1}
		\frac{a_n}{b_n}&=&n\frac{a_nf(a_n)}{\sqrt{i}} \\ \nonumber
							&=&\frac{i}{1-F(a_n)}\frac{a_nf(a_n)}{\sqrt{i}} \\ \nonumber
							&=&\sqrt{i}a_nh(a_n),
	\end{eqnarray}
where $h(x)=\frac{f(x)}{1-F(x)}$.

If we have the condition $\lim_{x \rightarrow \infty} x h(x) =c_0>0$, then it is clear that
	\begin{equation}\label{Eq_Proof_Lemma_an_bn_2}
		\lim_{n \rightarrow \infty} \frac{a_n}{b_n}=\infty.
	\end{equation}
If we have the other condition $\lim_{x \rightarrow \infty} \frac{d}{dx}\frac{1}{h(x)} =0$, we will have
\begin{eqnarray}\label{Eq_Proof_Lemma_an_bn_3}
	\lim_{n \rightarrow \infty} \frac{a_n}{b_n}&=&\lim_{n \rightarrow \infty} \sqrt{i}a_nh(a_n) \\ \nonumber
                        &\geq&\lim_{a_n \rightarrow \infty} \frac{a_n}{\frac{1}{h(a_n)}} \\ \nonumber
						&\stackrel{L'H\hat{o}pital }{=}& \lim_{a_n \rightarrow \infty} \frac{1}{\left(\frac{1}{h(a_n)}\right)'} \\ \nonumber
						&=& \infty,
\end{eqnarray}
where we have used the $L'H\hat{o}pital$'s rule.
\end{proof}

The next lemma is a simple result of the large deviations theory for distributions of NHT type:

\begin{lem}\label{Lemma_Lower_LDP_NHT}
Assume that $X_1,X_2,\dots,X_n$ are i.i.d. random variables of NHT type (as defined in definition \ref{Def_Condition_NHT_RV}). Then, there exists $n_0$ such that for all $n>n_0$ we have
\begin{equation}\label{Eq_Lemma_LDP_NHT}
	\Pr \left\{ \frac{X_1+\dots+X_{n}}{n} \geq K\mu \right\} < c_1 e^{-c_2n},
\end{equation}
for some strictly positive constants $c_1$ and $c_2$. Also $K>1$, and $\mathbb{E}\{X_1\}=\mu$ .
\end{lem}

\begin{proof}[Proof of Lemma \ref{Lemma_Lower_LDP_NHT}]
The proof of Lemma \ref{Lemma_Lower_LDP_NHT} can be found in \cite{Durrett} (see also \cite{Hollander} Ch. 1).
\end{proof}

The next lemma is a result from the large deviations theory for the distributions of HT type:

\begin{lem}\label{Lemma_Lower_LDP_HT}
Suppose the random variable $X$ (with $\mathbb{E}\{X\}=\mu$)  is of HT type (as defined in definition \ref{Def_Condition_HT_RV}). Then, for the sum of $n$ i.i.d. such random variables we will have:
\begin{equation}\label{Eq_Lemma_LDP_NHT}
\Pr\{ X_1 +\dots+ X_{n} >(1+\delta)n\mu\}\sim n\Pr\{ X_1>  n \mu \delta\},
\end{equation}
where $\delta>0$ is arbitrarily small, and $\sim$ stands for asymptotic equivalence\footnote{For details refer to \cite{Mikosch}.}.
\end{lem}
\begin{proof}[Proof of Lemma \ref{Lemma_Lower_LDP_HT}]
The proof of Lemma \ref{Lemma_Lower_LDP_HT} can be found in \cite{Mikosch}.
\end{proof}

Next, we provide the proof of the theorem.

\begin{proof} [Proof of Theorem \ref{Th_Main_Theorem_Lower_Bound}]
\begin{itemize}
\item Proof Strategy
\end{itemize}
Define $t\triangleq G^{-1}(n)$, where $G(.)$ is defined in (\ref{Eq_Main_Theorem_Lower_Bound_G_X}) and $n$ is the number of source-destination pairs. We prove that it is possible to have $t_1 \triangleq (1-\epsilon)t^{1-\delta}$ concurrent successful transmissions with high probability, where $\epsilon$ and $\delta$ are arbitrarily small positive constants. In order to prove this, at each time slot, we propose a subset of nodes with $(1-\epsilon)t^{1-\delta}$ members as a candidate for the set of active nodes $\mathbb{S}$. The proposed candidate set of active nodes consists  of the nodes with the best direct link power, as stated in (\ref{Eq_Th1_Proof_Idea_Candidate_Set}) where the main idea behind the proof was explained ($t_1$ indicated in (\ref{Eq_Th1_Proof_Idea_Candidate_Set}) is set equal to $(1-\epsilon)t^{1-\delta}$ in this proof) . Then, we prove that all destinations of nodes in $\mathbb{S}$ will satisfy the constraint $SINR>\beta$, with high probability, where $\beta$ is the $SINR$ target.

\begin{itemize}
\item[o] Power of the Desired Signal at Each Destination
\end{itemize}
Define the followings:
\begin{eqnarray}\label{Eq_Main_Theorem_Lower_Bound_Proof_Definitions}
	r_1 &\triangleq& n-t_1+1, \\ \nonumber
	\phi &\triangleq& K\mu t_1,
\end{eqnarray}
where $K>1$ is constant. Then, $\gamma_{(r_1),(r_1)}$ is the power of the desired signal at the weakest direct link among active sources. In order to prove that all sources will have successful transmission, we have to analyze the statistical properties of $\gamma_{(r_1),(r_1)}$ for which we will use Lemma \ref{Lemma_Lower_Falk}. From the construction of the set of candidate active nodes, it is clear that the involved random variables in our problem are $\gamma_{i,i}$'s, and we need to investigate their order statistics $\gamma_{(i),(i)}$'s. Thus, to use Lemma \ref{Lemma_Lower_Falk} we set:
\begin{equation}\label{Eq_Main_Theorem_Lower_Bound_Proof_Lemma_Match}
	X_{(i)}=\gamma_{(i),(i)},
\end{equation}
for $i=1,\dots,n$. Also, we have for the corresponding sequences
\begin{eqnarray}\label{Eq_Main_Theorem_Lower_Bound_Proof_An_Bn}
	a_n&=&F^{-1}\left(1-\frac{t_1}{n}\right), \\ \nonumber
	b_n&=&\frac{\sqrt{t_1}}{nf(a_n)}.
\end{eqnarray}
Then, we have:
\begin{eqnarray}\label{Eq_Main_Theorem_Lower_Bound_Proof_Main_Inequalities}
	\lim_{n \rightarrow \infty} \Pr \left\{ \gamma_{(r_1),(r_1)} > \beta \phi \right\} &=& \lim_{n \rightarrow \infty} \Pr \left\{ \gamma_{(r_1),(r_1)} > \beta K\mu (1-\epsilon)t^{1-\delta}\right\} \\ \nonumber
	&\stackrel{(a)}{\geq}& \lim_{n \rightarrow \infty} \Pr \left\{ \gamma_{(r_1),(r_1)} >  (1-\epsilon)\frac{t\beta\mu}{2}\right\} \\ \nonumber
	&\stackrel{(b)}{=}& \lim_{n \rightarrow \infty} \Pr \left\{ \gamma_{(r_1),(r_1)} >  (1-\epsilon)F^{-1}\left(1-\frac{t}{n}\right)\right\} \\ \nonumber
	&\stackrel{(c)}{\geq}& \lim_{n \rightarrow \infty} \Pr \left\{ \gamma_{(r_1),(r_1)} >  (1-\epsilon)F^{-1}\left(1-\frac{t_1}{n}\right)\right\} \\ \nonumber
	&\stackrel{(d)}=& \lim_{n \rightarrow \infty} \Pr \left\{ \gamma_{(r_1),(r_1)} >  (1-\epsilon)a_n\right\} \\ \nonumber
	&=& \lim_{n \rightarrow \infty} \Pr \left\{ \frac{\gamma_{(r_1),(r_1)} -a_n}{b_n} >  -\epsilon\frac{a_n}{b_n}\right\} \\ \nonumber
	&\stackrel{(e)}{=}& 1.
\end{eqnarray}
Inequality (a) is valid for large-enough\footnote{By the term ``large-enough'' we mean that there exists a constant $t_0$, independent of $n$, such that for all $t>t_0$ this fact holds. Thus, inequality (a) is valid for large-enough $t$, since, $t^{-\delta}$ is less than any positive constant for large-enough $t$.} $t$ . Equality (b) is due to the fact that we have set $t=G^{-1}(n)$, which results in $(t\beta\mu)/2=F^{-1}(1-t/n)$. Inequality (c) uses the fact that $F^{-1}(x)$ is an increasing function and the fact that $t_1=(1-\epsilon)t^{1-\delta}<t$. In step (d), we have used (\ref{Eq_Main_Theorem_Lower_Bound_Proof_An_Bn}). Finally, equality (e) is due to Lemma \ref{Lemma_Lower_Falk} and Lemma \ref{Lemma_Lower_an_bn}.

Thus, we have proved that:
\begin{eqnarray}\label{Eq_Main_Theorem_Lower_Bound_Proof_Desired_SINR_Condition}
	\lim_{n \rightarrow \infty} \Pr \left\{ (\gamma_{(r_1),(r_1)}>\beta \phi) \cap (\gamma_{(r_1+1),(r_1+1)}>\beta \phi) \cap \dots \cap  (\gamma_{(n),(n)}>\beta \phi) \right\}
 &=& \lim_{n \rightarrow \infty} \Pr \left\{ \gamma_{(r_1),(r_1)}>\beta\phi\right\} \\ \nonumber
 &=& 1,
\end{eqnarray}
which states that the power of the desired signal at all destinations are simultaneously above $\beta\phi$, with high probability.

\begin{itemize}
\item[o] Power of Unwanted Interference at Each Destination
\end{itemize}
The harmful interference for the transmission from $S_{(n-i+1)}$ to $D_{(n-i+1)}$ is denoted by $I_i$ (for $i=1,\dots,t_1$). Since while sorting the source-destination pairs we have not paid any attention to the cross-links, $I_i$'s are statistically equivalent. Then, we will have:
\begin{eqnarray}\label{Eq_Main_Theorem_Lower_Bound_Proof_Desired_Interference_Iequalities_NHT}
\Pr\{(I_1<\phi) \cap \dots \cap (I_{t_1}<\phi) \} &=& 1- \Pr\{(I_1\geq\phi) \cup \dots \cup (I_{t_1}\geq\phi)\} \\ \nonumber
&\geq& 1-\sum_{i=1}^{t_1}{\Pr\{I_i \geq \phi\}} \\ \nonumber
& \geq& 1-t_1c_1e^{-c_2t_1} \rightarrow 1 \hspace{1em} \rm as \it \hspace{1em} t_1 \rightarrow \infty,
\end{eqnarray}
where the last inequality follows from Lemma \ref{Lemma_Lower_LDP_NHT} by remembering that each $I_i$ is the sum of $t_1$ i.i.d. random variables with mean $\mu$. It should be noted that Lemma \ref{Lemma_Lower_LDP_NHT} only applies to distribution of NHT type. For distributions of HT type we use Lemma \ref{Lemma_Lower_LDP_HT} to get to the same conclusion:
\begin{eqnarray}\label{Eq_Main_Theorem_Lower_Bound_Proof_Desired_Interference_Iequalities_HT}
\Pr\{(I_1<\phi) \cap \dots \cap (I_{t_1}<\phi) \} &=& 1- \Pr\{(I_1\geq\phi) \cup \dots \cup (I_{t_1}\geq\phi)\} \\ \nonumber
&\geq& 1-\sum_{i=1}^{t_1}{\Pr\{I_i \geq \phi\}} \\ \nonumber
& =&1-t_1\Pr\{I_1 \geq \phi\}  \rightarrow 1 \hspace{1em} \rm as \it \hspace{1em} t_1 \rightarrow \infty,
\end{eqnarray}
which results from the following facts in the three categories of HT distributions:

\begin{itemize}

\item
Regularly varying tail:

\begin{equation}\label{Eq_Main_Theorem_Lower_Bound_Proof_Interference_Regularly_Varying_Tail}
    t_1\Pr\{I_1>\phi\}=t_1^2 \times \frac{L((K-1)\mu t_1)}{((K-1)\mu t_1)^\alpha}\rightarrow 0,
\end{equation}
for $\alpha>2$.

\item
Log-normal type tail:
\begin{equation}\label{Eq_Main_Theorem_Lower_Bound_Proof_Interference_LogNormal_Tail}
    t_1\Pr\{I_1>\phi\}\sim t_1^2 \times  c((K-1)\mu t_1)^\beta e^{-\lambda\log^\gamma ((K-1)\mu t_1)} \rightarrow 0,
\end{equation}
for $\gamma>1$ and $\lambda>0$.

\item
Weibull-like tail:
\begin{equation}\label{Eq_Main_Theorem_Lower_Bound_Proof_Interference_Weibull_Like_Tail}
    t_1\Pr\{I_1>\phi\}\sim t_1^2 \times  c((K-1)\mu t_1)^\beta e^{-\lambda ((K-1)\mu t_1)^\alpha} \rightarrow 0,
\end{equation}
for $0<\alpha<0.5$ and $\lambda>0$.

\end{itemize}

Thus, the interference at all destinations are simultaneously below $\phi$, with high probability (w.h.p.).

\begin{itemize}
\item[o] Signal to Interference and Noise Ratio
\end{itemize}
Finally, by considering Eq. (\ref{Eq_Main_Theorem_Lower_Bound_Proof_Desired_SINR_Condition}), Eq. (\ref{Eq_Main_Theorem_Lower_Bound_Proof_Desired_Interference_Iequalities_NHT}) and Eq. (\ref{Eq_Main_Theorem_Lower_Bound_Proof_Desired_Interference_Iequalities_HT}) -- for both NHT and HT type distributions -- we have proved that the desired signal power is greater than $\beta\phi$ at all the destinations, while, the interference is less than $\phi$ at all of them. Thus, we conclude that all the transmissions satisfy the $SINR$ constraint simultaneously (w.h.p.), which concludes the proof of the theorem.
\end{proof}

\section*{Appendix B: Proof of Corollaries \ref{Cor_Lower_Nakagami}, \ref{Cor_Lower_Weibull}, \ref{Cor_Lower_Pareto}, and \ref{Cor_Lower_Log_Normal}}\label{App2}
The proof of all four corollaries is done in three consecutive steps. For Corollary \ref{Cor_Lower_Nakagami} we have:

\begin{proof}[Proof of Corollary \ref{Cor_Lower_Nakagami}]
\begin{itemize}
\item Checking Conditions of the Theorem

First, we should verify that the Gamma distribution satisfies the theorem conditions. It is straightforward to show that the Gamma distribution satisfies the condition set 1, which is shown in detail in \cite{Aggarwal}. Also, since $\mathbb{E}\{e^{t x}\}=1/(1-\frac{\Omega}{m} t)^m$ \cite{Walck}, it satisfies the Cram\'{e}r's condition. Therefore it is of NHT type, and, we can apply Theorem \ref{Th_Main_Theorem_Lower_Bound} to Nakagami fading case.

\item Calculating the order of growth of $G^{-1}(x)$

In order to find the order of growth of $G^{-1}(x)$, first we should find the asymptotic expression of $G(x)$. In order to do that, we use the following facts:

First, we have lower and upper incomplete gamma functions as follows, respectively (\cite{Abramowitz}, Ch. 6):
\begin{equation}\label{Eq_Proof_Cor1_Lower_Incomplete_Gamma}
	\gamma(a,x) \triangleq \int_{0}^{x}{t^{a-1}e^{-t}dt},
\end{equation}
and
\begin{equation}\label{Eq_Proof_Cor1_Upper_Incomplete_Gamma}
	\Gamma(a,x) \triangleq \int_{x}^{\infty}{t^{a-1}e^{-t}dt}.
\end{equation}
Second, we know that (\cite{Abramowitz}, Ch. 6):
\begin{equation}\label{Eq_Proof_Cor1_Lower_Upper_Gamma_Sum}
	\gamma(a,x)+\Gamma(a,x)=\Gamma(a).
\end{equation}
Third, we have (\cite{Abramowitz}, Ch. 6):
\begin{equation}\label{Eq_Proof_Cor1_Upper_Gamma_Asymptotic}
	\frac{\Gamma(a,x)}{x^{a-1}e^{-x}} \rightarrow 1, \hspace{1em} \rm as \it \hspace{1em} x \rightarrow \infty.
\end{equation}
Finally, from relation (\ref{Eq_Main_Theorem_Lower_Bound_G_X}), by using the relations (\ref{Eq_Proof_Cor1_Lower_Incomplete_Gamma}) to (\ref{Eq_Proof_Cor1_Upper_Gamma_Asymptotic}), and with some calculations, one observes that
\begin{eqnarray}\label{Eq_Proof_Cor1_G_X_Asymptotic}
    G(x)&\stackrel{(a)}=&\frac{x}{1-F\left(\beta\mu x /2\right)} \\ \nonumber
        &\stackrel{(b)}=&\frac{x}{1-\frac{\gamma(m,m\beta x /2)}{\Gamma(m)}} \\ \nonumber
        &\stackrel{(c)}=&\frac{\Gamma(M) x}{\Gamma(m,m\beta x /2)} \\ \nonumber
	 &\stackrel{(d)}\rightarrow& \Gamma(m) \left(\frac{2}{\beta m}\right)^{m-1}x^{2-m}e^{\frac{\beta}{2}mx} \hspace{1em} \rm as \it \hspace{1em} x \rightarrow \infty,
\end{eqnarray}
where (a), (b), (c) and (d) are  due to (\ref{Eq_Main_Theorem_Lower_Bound_G_X}), (\ref{Eq_Cor_Nakagami_CDF}), (\ref{Eq_Proof_Cor1_Lower_Upper_Gamma_Sum}) and (\ref{Eq_Proof_Cor1_Upper_Gamma_Asymptotic}), respectively. Now, our next step is to calculate the order of growth of $G^{-1}(x)$. First, we observe that there exist positive constants $c_1$ and $c_2$ such that for large-enough $y$ we have:
\begin{equation}\label{Eq_Proof_Cor1_Inverse_G_X_Asymptotic_1}
    c_1\log G(y) <y <c_2 \log G(y)
\end{equation}
Then, setting $y=G^{-1}(x)$ yields the following expression for large-enough $x$\footnote{Note that $G(x)$ is an strictly increasing function.}:
\begin{equation}\label{Eq_Proof_Cor1_Inverse_G_X_Asymptotic_2}
    c_1\log x <G^{-1}(x) <c_2 \log x,
\end{equation}
which will result in
\begin{equation}\label{Eq_Proof_Cor1_Inverse_G_X_Asymptotic_2}
	G^{-1}(x) =\Theta\left( \log(x) \right).
\end{equation}

\item Achievable Throughput

Now we are ready to state the throughput result. By applying (\ref{Eq_Main_Theorem_Lower_Bound_General_Throughput}) we will have
\begin{eqnarray}\label{Eq_Proof_Cor1_Throughput}
	T&=&\Omega\left(G^{-1}(n)\right) \\ \nonumber
	&=&\Omega\left(\log(n)\right).
\end{eqnarray}
\end{itemize}
It is important to note that the case of $m=1$ will result in the Rayleigh fading environment.
\end{proof}
As in the previous corollary, we prove Corollary \ref{Cor_Lower_Weibull} through the following three consecutive steps:

\begin{proof}[Proof of Corollary \ref{Cor_Lower_Weibull}]
\begin{itemize}
\item Checking Conditions of the Theorem

It is easy to check that Weibull distribution satisfies condition set 1, which can be found in detail in \cite{Aggarwal}. Subsequently, if we have $k \geq 1$, then the Weibull distribution will be of NHT (super-exponential) type, while if $0<k<0.5$, the Weibull distribution satisfies the conditions required for the HT (sub-exponential) type, as given in Definition \ref{Def_Condition_HT_RV}. Thus, this corollary holds for $k\in (0,0.5) \cup [1,\infty)$.

\item Calculating the order of growth of $G^{-1}(x)$

It is easy to see that
\begin{equation}\label{Eq_Proof_Cor2_G_X_Asymptotic}
	G(x)=xe^{(\beta \Gamma(1+1/k) x /2)^k}.
\end{equation}
Now, our next step is to calculate the order of growth of $G^{-1}(x)$. First, we observe that there exist positive constants $c_1$ and $c_2$ such that for large-enough $y$ we have:
\begin{equation}\label{Eq_Proof_Cor2_Inverse_G_X_Asymptotic_1}
    c_1(\log G(y))^{1/k} <y <c_2 (\log G(y))^{1/k}.
\end{equation}
Then, setting $y=G^{-1}(x)$ yields the following expression for large-enough $x$:
\begin{equation}\label{Eq_Proof_Cor2_Inverse_G_X_Asymptotic_2}
    c_1(\log x)^{1/k} <G^{-1}(x) <c_2 (\log x)^{1/k},
\end{equation}
which will result in
\begin{equation}\label{Eq_Proof_Cor2_Inverse_G_X_Asymptotic_3}
	G^{-1}(x) =\Theta\left( (\log(x))^{1/k} \right).
\end{equation}

\item Achievable Throughput

Finally, by applying (\ref{Eq_Main_Theorem_Lower_Bound_General_Throughput}) we will have
\begin{equation}\label{Eq_Proof_Cor2_Throughput}
 T=\Omega\left(\left(\log(n)\right)^{1/k}\right).
\end{equation}
\end{itemize}
\end{proof}
Like before, for Corollary \ref{Cor_Lower_Pareto} we will follow three steps:

\begin{proof}[Proof of Corollary \ref{Cor_Lower_Pareto}]
\begin{itemize}

\item Checking Conditions of the Theorem

It can be easily checked that the Generalized Pareto distribution satisfies $\lim_{x \rightarrow \infty} x h(x) =c_0 >0$, and also is of HT type in the class of regularly varying tail distributions. Thus, it satisfies the theorem conditions.

\item Calculating the order of growth of $G^{-1}(x)$

It is easy to observe that
\begin{equation}\label{Eq_Proof_Cor3_Inverse_G_X_Asymptotic}
    G^{-1}(x)=\Theta\left(x^{1/(1+\alpha)}\right).
\end{equation}

\item Achievable Throughput

Finally, by applying (\ref{Eq_Main_Theorem_Lower_Bound_General_Throughput}) we will have
\begin{eqnarray}\label{Eq_Proof_Cor3_Throughput}
	T&=&\Omega\left(G^{-1}(n)\right) \\ \nonumber
	&=&\Omega\left(n^{1/(1+\alpha)}\right).
\end{eqnarray}
\end{itemize}
\end{proof}
Finally, for Corollary \ref{Cor_Lower_Log_Normal} we have:

\begin{proof}[Proof of Corollary \ref{Cor_Lower_Log_Normal}]
\begin{itemize}

\item Checking Conditions of the Theorem

It can be easily checked that the Log-normal distribution satisfies the condition set 1, which is provided in detail in \cite{Aggarwal}. Also, it is easy to note that it is of HT type.

\item Calculating the order of growth of $G^{-1}(x)$

First, in order to manipulate $G(x)$, we use the asymptotic expansion of the error function as follows:
\begin{equation}\label{Eq_Proof_Cor4_F_X_Asymptotic}
    1-F(x) \rightarrow \frac{\sigma}{\sqrt{2\pi}}\frac{e^{-(\log x -\mu)^2/2\sigma^2}}{\log x -\mu}, \hspace{1em} \rm as \it \hspace{1em} x \rightarrow \infty.
\end{equation}

Now, our next step is to calculate the order of growth of $G^{-1}(x)$. First we observe that there exist positive constants $c_1$ and $c_2$ such that for large-enough $y$ we have:
\begin{equation}\label{Eq_Proof_Cor4_Inverse_G_X_Asymptotic_1}
    c_1e^{\sqrt{2}\sigma\sqrt{\log(G(y))}} <y <c_2 e^{\sqrt{2}\sigma\sqrt{\log(G(y))}}.
\end{equation}
Then, setting $y=G^{-1}(x)$ yields the following expression for large-enough $x$:
\begin{equation}\label{Eq_Proof_Cor4_Inverse_G_X_Asymptotic_2}
    c_1e^{\sqrt{2}\sigma\sqrt{\log(x)}} <G^{-1}(x) <c_2 e^{\sqrt{2}\sigma\sqrt{\log(x)}},
\end{equation}
which will result in
\begin{equation}\label{Eq_Proof_Cor4_Inverse_G_X_Asymptotic_3}
    G^{-1}(x)=\Theta\left(e^{\sqrt{2}\sigma\sqrt{\log(x)}}\right).
\end{equation}

\item Achievable Throughput

Thus, we will have
\begin{eqnarray}\label{Eq_Proof_Cor4_Throughput}
	T&=&\Omega\left(G^{-1}(n)\right) \\ \nonumber
	&=&\Omega\left(e^{\sqrt{2}\sigma\sqrt{\log(n)}}\right).
\end{eqnarray}

\end{itemize}
\end{proof}

\section*{Appendix C: Proof of Theorem \ref{Th_Upper_NHT}}\label{App3}

Before presenting the proof of Theorem \ref{Th_Upper_NHT} we need some lemmas. The first lemma is an application of the large deviations theorem:

\begin{lem}\label{Lemma_Upper_Bound_Proof_LDP_NHT}
 If  $X_1,\dots,X_t$ are i.i.d. random variables, with the average of $\mu$, satisfying Cram\'{e}r's condition, then there exist $t_0$, $c_3>0$ and $c_4>0$ such that
\begin{equation}\label{Eq_Lemma_Upper_Bound_Proof_LDP_NHT}
	 \Pr\left\{ \frac{X_1+\dots+X_t}{t}<\frac{\mu}{2}\right\} <c_3e^{-c_4t},
\end{equation}
for all $t>t_0$.
\end{lem}

\begin{proof}[Proof of Lemma \ref{Lemma_Upper_Bound_Proof_LDP_NHT} ]
The proof of Lemma \ref{Lemma_Upper_Bound_Proof_LDP_NHT} can be found in \cite{Hollander}, Ch. 1.
\end{proof}

\begin{lem}\label{Lemma_Upper_Bound_Proof_ABC}
	Consider random variables $A$ and $B$. Then, we have the following
	\begin{equation}\label{Eq_Lemma_Upper_Bound_Proof_ABC}
		\Pr \{A>B\} \leq \Pr\{A>C\} +\Pr\{B<C\},
	\end{equation}
	for any arbitrary random variable $C$.
\end{lem}

\begin{proof}[Proof of Lemma \ref{Lemma_Upper_Bound_Proof_ABC} ]
For two random variables $A$ and $B$ and any random variable $C$, we have
\begin{eqnarray}\label{Eq_Lemma_Proof_Upper_Bound_Proof_ABC_1}
	A>B &\rightarrow& \overline{\left(A\leq C \cap B\geq C\right)} \\ \nonumber
		&=& (A > C) \cup (B<C),
\end{eqnarray}
where $\overline{D}$ is the complement of the event $D$. Thus,
	\begin{equation}\label{Eq_Lemma_Proof_Upper_Bound_Proof_ABC_1}
		\Pr \{A>B\} \leq \Pr\{A>C\} +\Pr\{B<C\}.
	\end{equation}
\end{proof}

The final necessary lemma (also used in \cite{Cui}) is:

\begin{lem}\label{Lemma_Upper_Bound_Proof_Choose}
For large $t$ and $n$ values we have
\begin{equation}\label{Eq_Lemma_Upper_Bound_Proof_Choose}
	{{n}\choose{t}} < \left(\frac{ne}{t}\right)^t.
\end{equation}
\end{lem}

\begin{proof}[Proof of Lemma \ref{Lemma_Upper_Bound_Proof_Choose}]
\begin{eqnarray}\label{Eq_Lemma_Proof_Upper_Bound_Proof_Choose}
	{{n}\choose{t}} &<&\frac{n^t}{t!} \\ \nonumber
							&\stackrel{(a)}{\approx}& \frac{n^t}{\sqrt{2\pi t} \left(\frac{t}{e}\right)^t} \\ \nonumber
							&<& \left(\frac{ne}{t}\right)^t,
\end{eqnarray}
where (a), which is valid for large values of $t$, is due to the Stirling's approximation.
\end{proof}

Based on the above three lemmas, we are ready to prove Theorem \ref{Th_Upper_NHT}.

\begin{proof}[Proof of Theorem \ref{Th_Upper_NHT}]
The proof strategy is in essence similar to the one stated in \cite{Cui} with critical generalizations to the super-exponential fading distribution. Define $t\triangleq c_7\log(n)$ for some positive constant $c_7$. In a communication scheme, we call the set $\mathbb{S}$ of the sources valid, if and only if, when they constitute the active set, all the destinations of the nodes in $\mathbb{S}$ are successful in decoding the message. It is clear that one can find many possible valid sets of sources. Define $X(t)$ to be the number of valid sets which have $t$ nodes. Then, according to Markov's inequality we will have:
\begin{eqnarray}\label{Eq_Theorem_Upper_Proof_1}
	\Pr\left\{X(t) \geq 1 \right\} &\leq& \mathbb{E} \left\{X(t) \right\} \\ \nonumber
							&=& {{n}\choose{t}} \left(\Pr\left\{A\right\}\right)^t,
\end{eqnarray}
in which $A$ is the event that a node can decode its message. Thus, we have
\begin{equation}\label{Eq_Theorem_Upper_Proof_2}
	\Pr\left\{ A \right\} =\Pr\left\{ \gamma_{i,i}> \beta(N_0+\sum_{j \in \mathbb{S}, j \neq i}{\gamma_{j,i}}) \right\},
\end{equation}
for any arbitrary $i \in \mathbb{S}$. Therefore,
\begin{eqnarray}\label{Eq_Theorem_Upper_Proof_3}
	\Pr\left\{X(t) \geq 1 \right\} &\leq&  {{n}\choose{t}} \left(\Pr\left\{A\right\}\right)^t \\ \nonumber
						&\leq & {{n}\choose{t}} \left(\Pr\left\{ \gamma_{i,i}> \beta\sum_{j \in \mathbb{S}, j \neq i}{\gamma_{j,i}} \right\}\right)^t \\ \nonumber
								& \stackrel{(a)}{\leq} &  {{n}\choose{t}} \left(\Pr\left\{ \frac{\gamma_{i,i}}{\beta}> \frac{\mu t}{2} \right\} + \Pr\left\{\sum_{j \in \mathbb{S}, j \neq i}{\gamma_{j,i}}<\frac{\mu t}{2}\right\} \right)^t \\ \nonumber
							& \stackrel{(b)}{<} &  {{n}\choose{t}} \left(\Pr\left\{ \frac{\gamma_{i,i}}{\beta}> \frac{\mu t}{2}\right\} + c_3e^{-c_4t}\right)^t \\ \nonumber
						&=&  {{n}\choose{t}} \left(1-F\left(\frac{\beta \mu t}{2}\right) + c_3e^{-c_4t}\right)^t \\ \nonumber
						&\stackrel{(c)}{<}&  {{n}\choose{t}} \left(c_5e^{-c_6t}\right)^t \\ \nonumber
						&\stackrel{(d)}{<}& \left(\frac{ne}{t}\right)^t \left(c_5e^{-c_6t}\right)^t \\ \nonumber
						&\stackrel{(e)}{\rightarrow}& 0.
\end{eqnarray}
Inequality (a) follows from Lemma \ref{Lemma_Upper_Bound_Proof_ABC} (by setting $C=\mu t /2$). Inequality (b) follows from Lemma \ref{Lemma_Upper_Bound_Proof_LDP_NHT}. Inequality (c) holds since the tail of $1-F(x)$ for a super-exponential distribution can be bounded from above by an exponentially decaying function, and the sum of two exponentially decaying functions can be bounded from above by another exponentially decaying function, for large values of their argument\footnote{These bounding discussions are applicable by using -- appropriately chosen -- constant parameters in each function.}. Inequality (d) follows from Lemma \ref{Lemma_Upper_Bound_Proof_Choose}. Finally, (e) follows by setting $t = (1/c_6)\log(n)$. Thus, we have shown that (for $t=\frac{1}{c_6}\log(n)$)
\begin{equation}\label{Eq_Theorem_Upper_Proof_4}
	\lim_{t\rightarrow \infty}\Pr\left\{X(t) \geq 1 \right\} =0,
\end{equation}
which states that with high probability we cannot find any set of $t=c_7\log(n)$ sources so that all will be successful. This concludes the proof.
\end{proof}

\bibliographystyle{ieeetr}

\end{document}